\documentclass[draftcls,onecolumn,11pt]{IEEEtran}

\usepackage{amsmath,amstext,amssymb,epsf}
\usepackage{fullpage,latexsym}
\usepackage{epsfig}
\usepackage{setspace}
\numberwithin{equation}{section} % in amsmath

 \newtheorem{lemma}{Lemma}[section]
 \newtheorem{theorem}[lemma]{Theorem}

 \newtheorem{claim}[lemma]{Claim}

 \newtheorem{rem}[lemma]{Remark}

 \newtheorem{ex}[lemma]{Example}

\begin{document}

\title{Nonapproximablity of the Normalized Information Distance}
\author{Sebastiaan A. Terwijn\thanks{SAT is with the University of Amsterdam;
Email terwijn@logic.at}, Leen Torenvliet\thanks{LT is with the University 
of Amsterdam; Email leen@science.uva.nl}, and Paul M.B. Vit\'{a}nyi\thanks{
PMBV is with the CWI and the University of Amsterdam; Email Paul.Vitanyi@cwi.nl}}
\maketitle

\begin{abstract}
Normalized information distance (NID) 
uses the theoretical notion of Kolmogorov complexity, which
for practical purposes is approximated by the length
of the compressed version of the file involved, using
a real-world compression program. This practical application is 
called `normalized compression distance' and it is trivially computable.
It is a parameter-free similarity measure
based on compression, and is used in pattern recognition, data mining,
phylogeny, clustering, and classification.
The complexity properties of its theoretical precursor, the NID, have been
open. We show that the NID is neither upper
semicomputable nor lower semicomputable up to any reasonable precision.

\end{abstract}

\section{Introduction}
\label{sect.intro}

The classical notion of
Kolmogorov complexity \cite{Ko65} is an objective measure
for the information in 
a {\em single} object, and information distance measures the information
between a {\em pair} of objects \cite{BGLVZ98}. This last notion has
spawned research in the theoretical direction, among others
\cite{CMRSV02,VV02,Vy02,Vy03,MV01,SV02}.
Research in the practical direction has focused on
the normalized information distance (NID), also called the similarity metric, which arises
by normalizing the information distance in a proper manner. If we also
approximate the Kolmogorov complexity through real-world
compressors \cite{Li03,CVW03,CV04}, 
then we obtain the
normalized compression distance (NCD). 
This is a parameter-free, feature-free, and alignment-free
similarity measure  that
has had great
impact in applications.
%A variant of this compression distance has been tested
%on all time sequence databases used in the last decade in the major
%data mining conferences (sigkdd, sigmod, icdm, icde, ssdb, vldb, pkdd, pakdd)
The NCD was preceded by a related 
nonoptimal distance \cite{LBCKKZ01}.
In \cite{KLRWHL07} another variant of the NCD has been tested
on all major time-sequence databases used in all major data-mining conferences
against all other major methods used. The compression method turned out
to be competitive in general 
and superior in
heterogeneous data clustering and anomaly detection.
There have been many applications in pattern recognition, phylogeny, clustering,
and classification, ranging from
hurricane forecasting and music to
to genomics and analysis of network traffic,
see the many 
papers referencing
\cite{Li03,CVW03,CV04}
in Google Scholar. 
The NCD is trivially computable.
In \cite{Li03} it is shown that its theoretical precursor, the NID, is 
a metric up to negligible discrepancies
in the metric (in)equalities and that it is always between 0 and 1.

The computability status of the NID has been open,
see Remark VI.1 in \cite{Li03} which asks whether the NID is 
upper semicomputable, and (open) Exercise 8.4.4 (c) in the textbook
\cite{LiVi08} which asks whether the NID is semicomputable at all. 
We resolve this question by showing the following.
\begin{theorem}
Let $x,y$ be binary strings of length $n$ and denote the NID between them 
by $e(x,y)$.

(i) There is no upper semicomputable function $g$ such that
$|g(x,y) - e(x,y)| \leq (\log n)/ n$ (Lemma~\ref{lem.upper}).

(ii) There is no lower semicomputable function $g$ such that
$|g(x,y) - e(x,y)| \leq 1/2$ (Lemma~\ref{lem.lower}).
\end{theorem}

\section{Preliminaries}

We write {\em string} to mean a finite binary string,
and $\epsilon$ denotes the empty string.
The {\em length} of a string $x$ (the number of bits in it)
is denoted by $|x|$. Thus,
$|\epsilon| = 0$.
Moreover, we identify strings with natural numbers  
by associating each string with its index
in the length-increasing lexicographic ordering
\[
( \epsilon , 0),  (0,1),  (1,2), (00,3), (01,4), (10,5), (11,6), 
\ldots . 
\]
Informally, the Kolmogorov complexity of a string
is the length of the shortest string from which the original string
can be losslessly reconstructed by an effective
general-purpose computer such as a particular universal Turing machine $U$,
\cite{Ko65}.
Hence it constitutes a lower bound on how far a
lossless compression program can compress.
In this paper we require that the set of programs of $U$ is prefix free
(no program is a proper prefix of another program), that is, we deal
with the {\em prefix Kolmogorov complexity}. 
(But for the results in this paper it does not matter whether we use
the plain Kolmogorov complexity or the prefix Kolmogorov complexity.)
We call $U$ the {\em reference universal Turing machine}.
Formally, the {\em conditional prefix Kolmogorov complexity}
$K(x|y)$ is the length of the shortest input $z$
such that the reference universal Turing machine $U$ on input $z$ with
auxiliary information $y$ outputs $x$. The
{\em unconditional prefix Kolmogorov complexity} $K(x)$ is defined by
$K(x|\epsilon)$.
For an introduction to the definitions and notions
of Kolmogorov complexity (algorithmic information theory)
see ~\cite{LiVi08}.

Let $\cal N$ and $\cal R$ denote
the nonnegative integers and the real numbers,
respectively. 
A function $f: {\cal N} \rightarrow {\cal R}$ is
{\em upper semicomputable} (or $\Pi_1^0$)
if it is defined by a rational-valued computable
function  $\phi (x,k)$ where $x$ is a string
and $k$ is a nonnegative integer
such that $\phi(x,k+1) \leq \phi(x,k)$ for every $k$ and
  $\lim_{k \rightarrow \infty} \phi (x,k)=f(x)$.
This means
  that $f$ can be computably approximated from above.
A function $f$ is
{\em lower semicomputable} (or $\Sigma_1^0$) 
  if $-f$ is upper semicomputable.
 A function is called
{\em semicomputable} (or $\Pi_1^0 \bigcup \Sigma_1^0$)
  if it is either upper semicomputable or lower semicomputable or both.
A function
$f$ is {\em computable}
(or recursive)
  iff it is both upper semicomputable and
lower semicomputable (or $\Pi_1^0 \bigcap \Sigma_1^0$).
Use $ \langle \cdot \rangle$
as a {\em pairing
function}
over ${\cal N}$ to associate a unique natural number $\langle x, y \rangle$
with each pair $(x, y)$ of natural numbers.
An example is $\langle x, y \rangle$ defined by
$y+(x+y+1)(x+y)/2$.
In this way we can extend the above definitions to functions of two
nonnegative integers, in particular to distance functions.

The {\em information distance} $D(x,y)$ between strings $x$ and $y$
is defined as 
\[
D(x,y)= \min_{p} \{|p|: U(p,x)=y \wedge U(p,y)=x \},
\]
where $U$ is the reference universal Turing machine above.
Like the Kolmogorov complexity $K$, the distance function $D$ 
is upper semicomputable. Define
\[E(x,y)= \max \{K(x|y),K(y|x)\}.
\]
In \cite{BGLVZ98} it is shown that
the function $E$ is upper semicomputable, 
$D(x,y)= E(x,y)+O(E(x,y))$, the function $E$ is a metric (more precisely,
that it satisfies the metric (in)equalities up to a constant),
and that $E$ is minimal (up to a constant) among all 
upper semicomputable distance functions $D'$ satisfying the mild
normalization conditions $\sum_{y:y \neq x} 2^{-D'(x,y)} \leq 1$ and
$\sum_{x:x \neq y} 2^{-D'(x,y)} \leq 1$ (the minimality property
was relaxed from metrics \cite{BGLVZ98} to symmetric distances
\cite{Li03} to the present form \cite{LiVi08} without serious
proof changes).
The {\em normalized information distance} $e$ is defined by
\[
e(x,y) = \frac{E(x,y)}{\max\{K(x),K(y)\}}.
\]
It is straightforward that $0 \leq e(x,y) \leq 1$ up to some minor
discrepancies for all $x,y \in \{0,1\}^*$.
Since $e$ is the ratio between two upper semicomputable functions,
that is, between two $\Pi_1^0$ functions, it is a $\Delta_2^0$ function.
That is, $e$  is computable relative to the halting problem $\emptyset '$.
One would not expect any better bound in the arithmetic hierarchy.
Call a function $f(x,y)$
{\em computable in the limit} if there exists
a rational-valued computable function $g(x,y,t)$ such that
 $\lim_{t \rightarrow \infty} g(x,y,t)$ $=f(x,y)$.
This is precisely the class of functions
that are Turing-reducible
to the halting set, and
the  NID
is in this class, Exercise~8.4.4 (b) in \cite{LiVi08} (a result due to
\cite{Ga01}).

\section{Nonapproximability of the NID from above}

\begin{lemma}\label{lem.upper}
There is no upper semicomputable function $g(x,y)$ such that
$|e(x,y)-g(x,y)| < (\log n)/n$ for all strings $x,y$ of length $n$.
\end{lemma}

\begin{proof}
%The proof of Terwijn and Torenvliet states that $e(x,x)$ and hence $e(x,y)$
%in general is not upper semicomputable. 
For simplicity
we use $e(x,x) = 1/K(x)$.

%Suppose we can compute a function $f(x)$ such that 
%\begin{equation}\label{eq.cond}
%|K(x)-f(x)| \leq n/\log n,
%\end{equation}
%where $n=|x|$.
%Therefore, given $x$ we can compute $f(x)$ and given at most $\log (n/\log n)$
%bits more we can compute $K(x)$. 
%(We can round $f(x)$ up or down and indicate the direction of rounding with one
%bit.)
%But this contradicts the following.
By Theorem 3.8.1 in \cite{LiVi08} (a result due to \cite{Ga74}),
for every $n$ there is an $x$ of length $n$ such that
\begin{equation}\label{eq.cond1}
K(K(x)|x) \geq \log \frac{n}{\log n} +O(1).
\end{equation}
%Hence there is no computable function $f$ such that \eqref{eq.cond} holds.

Assume there is an upper semicomputable function $g(x)$ such
that 
\[
|e(x,x) - g(x)| < \frac{\log n}{n },
\]
with $n=|x|$. That is, 
$|1/K(x) - g(x)| < (\log n)/n $. Then, 
$1/K(x)$ is upper semicomputable 
within distance $(\log n)/n$. Therefore, $K(x)$ is lower semicomputable
within distance  $n/\log n$. Since $K(x)$ is also upper semicomputable, it
is computable within distance $n/\log n$. But this violates \eqref{eq.cond1}
because we can describe the given distance in $\log (n/\log n)$ bits. 
(We can round $g(x)$ up or down and indicate 
the direction of rounding with one
bit.)
\end{proof}

\section{Nonapproximability of the NID from below}

Let $x$ be a string of length $n$ and $t(n)$ 
a computable time bound. Then $K^t$ denotes the time bounded version
of $K$ defined by
\[
K^t(x) = \min_p \{|p|: U(p) = x \; \text{\rm in at most} \; t(n) \; 
\text{\rm steps}\}.
\]
The computation of $U$ is measured in terms of the output rather than
the input, which is more natural in the context of Kolmogorov complexity.
Define the
time bounded version $E^t$ of $E$ by
\[
E^t(x,y) = \max \{K^t(x|y), K^t(y|x) \}.
\]
\begin{lemma}\label{lem.1}
For every length $n$ and computable time bound $t(n)$ there are 
strings $x$ and $a$ of length $n$ such that
\begin{itemize}
\item $K(x) \geq n$,
\item $K(x|a) \geq n$,
\item $K(a|n)=O(1)$,
\item $K^t(a|x) \geq n - O(1)$.
\end{itemize}
\end{lemma}

\begin{proof}
See \cite{BT09}, Lemma 7.7.
\end{proof}

\begin{lemma}\label{lem.xor}
For every length $n$ and large enough computable time bound $t(n)$,
there exist strings $x$ and $y$ of length $n$ such that
\begin{itemize}
\item $K(x) \geq n$,
\item $E(x,y) = O(1)$,
\item $E^t(x,y) \geq n-O(1)$ (where the constant in the big-O depends on $t$
but not on $n$)
\end{itemize}
\end{lemma} 
\begin{proof}
Let $x$ and $a$ be as in Lemma~\ref{lem.1}, using $2t(n)$ instead of $t(n)$.
In this way, we have $K^{2t} (a|x) \geq n-O(1)$. Define $y$ by $y=x \oplus a$
where $\oplus$ denotes the bitwise XOR. Then,
\[
E(x,y) \leq K(a|n)+O(1)=O(1).
\]  
We also have $a=x \oplus y$ so that (with the time bound $t$ large enough)
\begin{eqnarray*}
n-O(1) & \leq & K^{2t} (a|x) 
\\& \leq & K^t(y|x) +O(1)
\\&\leq& \max \{K^t(x|y), K^t(y|x)\}+O(1)
\\&=& E^t(x,y)+O(1).
\end{eqnarray*} 
\end{proof}

\begin{lemma}\label{lem.lower}
There is no lower semicomputable function $g(x,y)$ such that
$|e(x,y)-g(x,y)| \leq \frac{1}{2}$ for all strings $x,y$.
\end{lemma}

\begin{proof}
Assume by way of contradiction that the lemma is false
for some strings $x,y$ of length $n$ with $x$ and $y$
satisfying the conditions in Lemma~\ref{lem.xor}. 
Let $g(x,y)+\delta=e(x,y)$ for some
$\delta$ with $- \frac{1}{2} \leq \delta \leq  \frac{1}{2}$.
Let $g_i$ be a lower semicomputable
function approximation of $g$ such that $g_{i+1}(x,y) \geq g_i(x,y)$ for
all $i$ and $\lim_{i \rightarrow \infty} g_i(x,y) = g(x,y)$. 
Let $E_i$ be an upper semicomputable function approximating $E$ such
that $E_{i+1}(x,y) \leq E_i(x,y)$ for all $i$ 
and $\lim_{i \rightarrow \infty} E_i(x,y) = E(x,y)$.
%Without loss of generality we can have $E_0(x,y) \leq n +O(\log n)$.
Finally, let  $s=s(x,y)$ be
the least $i$ such that
\[
g_s(x,y) + \delta \geq \frac{E_s(x,y)}{n+2 \log n +O(1)}.
\]
Here ``$\log$'' denotes the binary logarithm.
Since $K(z) \leq n+2 \log n +O(1)$ for every string $z$ of length $n$,
see \cite{LiVi08},
and $\lim_{i \rightarrow \infty}g_i(x,y) +\delta = e(x,y)$,
such an $s$ exists by the contradictory assumption.
\begin{claim}\label{claim.es}
There is a constant $c$  depending on $s$ but
not on $x$ and $y$ such that $E_s (x,y))\geq n-c$.
\end{claim}
\begin{proof}
Define a computable time bound $t(n)$ such that
\begin{equation}\label{eq.st}
E_s (u,v)) < n-c \Longrightarrow E^t (u,v) < n-c,
\end{equation}
for all strings $u,v$ of length $n$, and $c$
a constant to be determined below, as follows. 
If $E_s(u,v) < n-c$, then $E(u,v) < n-c$. There is a program
witnessing $E_s(u,v) < n-c$, and we can define $\hat{t}(u,v)$
to be the running time of this program. 
Let $t(n)$
be the maximum of the $\hat{t}(u,v)$'s for all pairs 
$u,v$ such that $E_s(u,v) < n-c$ with $s$ as in \eqref{eq.st}.
By Lemma~\ref{lem.xor}, we have $E^t(x,y) \geq n-c'$ for some constant $c'$
depending on $t$ but not on $n$ (and $x$ and $y$). Set $c=c'$. 
Then, by \eqref{eq.st}
we have $E_s (x,y)) \geq n-c$. 
\end{proof}
Recall that $e(x,y)=g(x,y)+\delta$. Consider the sequence of inequalities
\begin{eqnarray*}
E(x,y)- \delta n & \geq & g(x,y)\cdot n 
\\&\geq & g_s(x,y) \cdot n 
\\& \geq &\frac{E_s(x,y) \cdot n}{n+2 \log n +O(1)} - \delta n
\\&= & \Omega (n)- \delta n,
%\frac{1}{2} E_s(x,y),
\end{eqnarray*}
for $n$ large enough.
The first inequality holds since $K(x) \geq n$.
The last inequality holds since $E_s(x,y) \geq n-c$ 
by Claim~\ref{claim.es}, and
$n/(n+2 \log n +O(1)) > \frac{1}{2}$ for $n$ large enough. 
%and $-\frac{1}{3} \leq \delta \leq \frac{1}{3}$. 
Since we have shown that
$E(x,y) = \Omega (n)$, we contradict 
$E(x,y)=O(1)$ by Lemma~\ref{lem.xor}.
\end{proof}

\section{Open Problem}

A subset of $\mathcal{N}$ is
called $n$-computably enumerable ($n$-c.e.) if it is 
a Boolean combination of $n$ computably enumerable
sets. Thus, the $1$-c.e. sets are the computably enumerable
sets, the $2$-c.e. sets (also called d.c.e.) the differences
of two c.e. sets, and so on.
The $n$-c.e. sets are referred to as the
{\em difference hierarchy} over the c.e. sets.
This is an effective analog of a classical hierarchy
from descriptive set theory.
Note that a set is $n$-c.e. if it has a
computable approximation that changes at most $n$ times.

We can extend the notion of $n$-c.e. set to a notion that
measures the number of fluctuations of a function as follows:
For every $n\geq 1$,
call $f:\mathcal{N}\rightarrow \mathcal{R}$
{\em $n$-approximable} if there is a rational-valued
computable approximation $\phi$ such that
$\lim_k \phi(x,k) = f(x)$ and such that for every $x$,
the number of $k$'s such that
$\phi(x,k+1) - \phi(x,k)<0$ is bounded by~$n-1$.
That is, $n-1$ is a bound on the number of fluctuations
of the approximation.
Note that the $1$-approximable functions are precisely the
lower semicomputable ($\Sigma^0_1$) ones (zero fluctuations).
Also note that a set $A\subseteq\mathcal{N}$ is $n$-c.e.
if and only if the characteristic function of $A$
is $n$-approximable.

{\bf Conjecture} For every $n\geq 1$, the normalized information
distance $e$ is not $n$-approximable.

\end{document}